\documentclass[pra,aps,reprint,a4paper,superscriptaddress,longbibliography,floatfix,accepted=2021-09-06]{quantumarticle}

\pdfoutput=1

\usepackage{epsfig}
\usepackage{amsfonts}
\usepackage{amsmath}
\usepackage{amssymb}
\usepackage{amsthm}
\usepackage{xcolor}
\usepackage{multirow}
\usepackage[normalem]{ulem}
\newcommand{\stkout}[1]{\ifmmode\text{\sout{\ensuremath{#1}}}\else\sout{#1}\fi}
\usepackage{latexsym}
\usepackage{mathrsfs}
\usepackage{natbib}
\usepackage{verbatim}
\usepackage[T1]{fontenc}
\usepackage{gensymb}

\usepackage[colorlinks=true,linkcolor=blue,citecolor=magenta,urlcolor=blue]{hyperref}

\DeclareMathOperator{\Tr}{tr}
\newtheorem{lemma}{Lemma}
\newtheorem{result}{Result}

\newcommand{\ket}[1]{|#1\rangle}
\newcommand{\bra}[1]{\langle#1|}

\newcommand{\ketbra}[2]{|#1\rangle\langle#2|}

\newcommand{\ctm}{\cos\frac{\theta_m}{2}}
\newcommand{\stm}{\sin\frac{\theta_m}{2}}
\newcommand{\red}[1]{\textcolor{red}{#1}}

\begin{document}


\title{Noise-robust preparation contextuality shared between any number of observers via unsharp measurements}


\author{Hammad Anwer}
\affiliation{Department of Physics, Stockholm University, S-10691 Stockholm, Sweden}

\author{Natalie Wilson}
\affiliation{Department of Physics, Stockholm University, S-10691 Stockholm, Sweden}

\author{Ralph Silva}
\affiliation{Institute for Theoretical Physics, ETH Zurich, Switzerland}

\author{Sadiq Muhammad}
\affiliation{Department of Physics, Stockholm University, S-10691 Stockholm, Sweden}

\author{Armin Tavakoli}
\affiliation{D\'epartement de Physique Appliqu\'ee, Universit\'e de Gen\`eve, CH-1211 Gen\`eve, Switzerland}
\affiliation{Institute for Quantum Optics and Quantum Information - IQOQI Vienna, Austrian Academy of Sciences, Boltzmanngasse 3, 1090 Vienna, Austria}
\affiliation{Institute for Atomic and Subatomic Physics, Vienna University of Technology, 1020 Vienna, Austria}

\author{Mohamed Bourennane}
\affiliation{Department of Physics, Stockholm University, S-10691 Stockholm, Sweden}

\begin{abstract}
Multiple observers who independently harvest nonclassical correlations from a single physical system share the system's ability to enable quantum correlations. We show that any number of independent observers can share the preparation contextual outcome statistics enabled by state ensembles in quantum theory. Furthermore, we show that even in the presence of any amount of white noise, there exists quantum ensembles that enable such shared preparation contextuality. The findings are experimentally realised by applying sequential unsharp measurements to an optical qubit ensemble which reveals three shared demonstrations of preparation contextuality.
\end{abstract}


\maketitle


\section{Introduction}
Quantum correlations can surpass the limitations of corresponding classical models. In their most common form, quantum correlations are obtained from the outcomes of \textit{single} (albeit randomly chosen) measurements performed on a physical system. After the measurement, the physical system can be discarded, or even demolished by the measurement apparatus.  Therefore, since one does not need to consider the measurement-induced decoherence in the state of the physical system, optimal quantum correlations are often obtained from sharp (projective) measurements that extract a maximal amount of information from the physical system while also inducing a maximal disturbance in its state \cite{Fuchs}.

Arguably, the fact that measurements disturb physical states should have interesting consequences for more general quantum correlations. To reveal the influence of measurement-induced disturbances on observed outcome statistics, one requires systems to undergo more than a single measurement. A simple scenario for studying the trade-off between the strength of quantum correlations and the disturbance induced by extracting them is one in which quantum correlations are \textit{shared} between many observers. Sharing quantum correlations means that a physical system is measured by a sequence of independent observers, each of whom are tasked with falsifying the existence of a classical model for their observed correlations. Hence, the stronger the correlations extracted by the first observer, the larger the disturbance induced in the state of the system, and thus the weaker the correlations that can possibly be extracted by a second observer. Sharing quantum correlations requires the first observer to measure in such a way that the outcome correlations are strong enough to elude all classical models while the induced disturbance is small enough to enable a second observer independently repeat the same feat. Understanding and characterising quantum correlations obtained via sequential measurements is a conceptually interesting problem \cite{SG16, BM13, GW14, TC18, Brown} which has promising applications in quantum information protocols \cite{CJ17, CH18, Foletto2020}. 

Sharing quantum correlations was first studied in the context of Bell inequality tests \cite{SG16} where it was found that a pair of qubits in a singlet state can enable two sequential Bell inequality violations. This has also been experimentally demonstrated \cite{Exp1, Exp2}. Moreover, shared quantum correlations have recently also been studied in other tasks such as entanglement witnessing \cite{BM18}, quantum steering \cite{SD18, SDH18} and a semi-device-independent setting \cite{Mohan, Miklin, AnwerRAC, FolettoRAC}.

Here, we theoretically and experimentally study the sharing of quantum correlations that demonstrate preparation contextuality. These are correlations that cannot be reproduced in a hidden variable theory that ascribes equivalent representations to indistinguishable preparations, i.e. it disregards the context (specific procedure) underlying a state preparation \cite{Sp05}. Such quantum contextuality does not require entanglement but only single quantum systems, and is well-studied both in theory   (see e.g.~Refs.\cite{Sp05, Negativity, SpekkensPOM, LM13, BB15, SC18, Ghorai, Uola, Hierarchy1, Hierarchy2}) and experiment (see e.g.~Refs.~\cite{SpekkensPOM, HT17, MP16}). In our scenario, states are sampled from an ensemble and communicated sequentially between independent observers, each of whom performs a measurement with the aim of obtaining preparation contextual outcome statistics. We show that preparation contextuality can be shared between any number of sequential observers. Furthermore, we show that the sharing is robust to noise, in the sense that for any given number of independent observers and exposure to any nontrivial amount of white noise, one can find an ensemble whose contextuality can be shared between all the observers. We proceed to experimentally demonstrate the sharing of preparation contextuality. We realise a four-observer scenario in which the first observer prepares an optical qubit ensemble and the remaining three observers perform sequential unsharp (non-maximally disturbing) measurements. Thus, we obtain three shared demonstrations of preparation contextuality.

\section{Nonclassicality via preparation contextuality}
The impossibility of describing the set of observables in quantum theory by underlying classical (noncontextual) quantities originates in the arguments of Bell, Kochen and Specker \cite{KochenSpecker}. More recently, the notion of contextuality has seen a generalisation  formulated in operational terms (i.e., in terms of probabilities) applying to measurements, transformations and preparations \cite{Sp05}. Here, we are interested in contextuality in terms of preparations.

The predictions of an operational theory (e.g. quantum theory) may be explained by an ontological model \cite{HS10}. An ontological model ascribes a set $\Lambda$ of ontic (objective) states $\lambda$ to each physical system S. A particular preparation $P$ of the system is  associated to a distribution $\mu_P(\lambda)$ over the ontic state space. Similarly, the probability of outcome $b$ of a measurement $M$ is described by a response function  $\xi_{b,M}(\lambda)$. The ontological model thus seeks a $\mu$ and a $\xi$ to explain the observed statistics by $p(b|P,M)=\int_\Lambda \mu_P(\lambda)\xi_{b,M}(\lambda)d\lambda$. Note that such an ontological model only concerns the observed probabilities and does not need to explicitely reference an underlying physical theory. The model is said to be \textit{preparation noncontextual} if two different preparations $P$ and $P'$ that cannot be distinguished by the statistics generated by any measurement (that is; $\forall M: p(b|P,M)=p(b|P',M)$) are associated to the \textit{same} distribution over ontic states, i.e., $\mu_P=\mu_{P'}$. If observed statistics falsify this assumption, the it is said to be preparation  \textit{contextual}. Quantum state ensembles are known to enable preparation contextuality. 

In order to prove preparation contextuality, it is sufficient to violate an inequality bounding the correlations attainable by any preparation noncontextual model. We focus on a family of such inequalities introduced in Ref.~\cite{SpekkensPOM} related to a variant of Random Access Coding \cite{Ambainis, RAC}. Consider a party Alice receiving a random input string $x=x_1\ldots x_n\in\{0,1\}^n$. Her input is associated to a preparation $P_x$ (one of $2^n$ possible) which is sent to a receiver Bob. Her preparations are constrained to satisfy certain indistinguishability relations: there must exist no measurement that can reveal any information about the parity of the string $r\cdot x$ for every $r\in\{0,1\}^n$ with $|r|\geq 2$.  Bob receives a random input $y\in\{1,\ldots,n\}$, and performs a measurement $\{M_y^b\}$ with outcome $b\in\{0,1\}$. The partnership is awarded a point if the outcome of Bob coincides with the $y$th entry in Alice's string. In any preparation noncontextual theory, the probability of winning obeys the following bound \cite{SpekkensPOM}:
\begin{equation}\label{ineq}
\mathcal{A}^{(n)}\equiv \frac{1}{n2^n}\sum_{x,y} p(b=x_y\lvert x,y)\leq \frac{n+1}{2n}.
\end{equation}
Due to the contextual nature of quantum theory, these inequalities can be violated. Maximal quantum violations for any $n\geq 2$ are known \cite{IndexGame}. Bob performs dichotomic measurements characterised by an observable $G_{n,y}^T$, where $T$ denotes transpose. These are recursively defined from $G_{2,1}=\sigma_x$, $G_{2,2}=\sigma_y$, and $G_{3,1}=\sigma_x$, $G_{3,2}=\sigma_y$ and $G_{3,3}=\sigma_z$, and
\begin{align}\nonumber\label{optmeas}
& \text{$n$ even:} & G_{n,k}=G_{n-1,k}\otimes \sigma_x \hspace{5mm}  \forall k\in\{1,\ldots, n-1\}, \\
& \text{$n$ odd:}  &  G_{n,k}=G_{n-2,k}\otimes \sigma_x  \hspace{5mm}  \forall k\in\{1,\ldots, n-2\}
\end{align}
with $ G_{n,n}=\openone\otimes \sigma_y$ if $n>3$ is even, and $ G_{n,n}=\openone\otimes \sigma_z$ and $ G_{n,n-1}=\openone\otimes \sigma_y$ if $n>3$ is odd. Note that the dimension of $G_{n,k}$ is $2^{\lfloor n/2\rfloor}$. The optimal preparations are states of $\lfloor n/2 \rfloor$ qubits specified by
\begin{equation}\label{optstate}
\rho_x=\Tr_\text{A}\left[\left(\openone + A_x\right)\otimes \openone \left(\ketbra{\phi_{\text{max}}}{\phi_{\text{max}}}\right)^{\otimes\lfloor n/2\rfloor}\right],
\end{equation}
where $A_x=\frac{1}{\sqrt{n}}\sum_{i=1}^{n}(-1)^{x_i}G_{n,i}$, $\ket{\phi_{\text{max}}}=\left(\ket{0,0}+\ket{1,1}\right)/\sqrt{2}$, and the trace is taken over the first system in every entangled pair. Note that Alice's preparations are single quantum systems, and only for simplicity written in terms of post-measurement states of a collection of entangled states. The presented strategy leads to the maximal quantum value $\mathcal{A}^{(n)}=1/2(1+1/\sqrt{n})$ for every $n$ \cite{IndexGame}.

\section{Sequential scenario} We consider a scenario in which the ability to violate the inequality \eqref{ineq} is shared between many independent observers, named Bob$_1$,..., Bob$_m$, each of whom receive an independent random input $y_k\in\{1,\ldots,n\}$ and output $b_k\in\{0,1\}$. Alice's randomly chosen preparation is sent to Bob$_1$ who performs a measurement and passes the post-measurement state to Bob$_{2}$ who performs a measurement and passes the post-measurement state to Bob$_3$ etc. The scenario is illustrated in Fig.~\ref{FigScenario}. The pair Alice-Bob$_k$ uses the marginal distribution $p(b_k|x,y_k)$ to compute the witness \eqref{ineq} (here labelled $\mathcal{A}_k^{(n)}$) to check for preparation contextuality.  
\begin{figure}
	\includegraphics[width=\columnwidth]{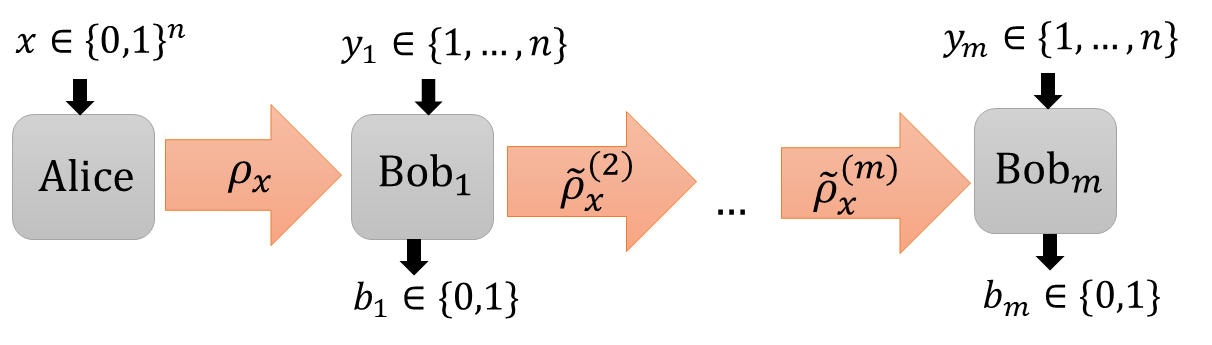}
	\caption{Alice's preparations are sent from one observer to the next, each performing a measurement aiming to independently reveal preparation contextual statistics. To this end, only the average post-measurement state $\tilde{\rho}_x^{(k)}$ is relevant.}\label{FigScenario}
\end{figure}

In a quantum approach, we may denote Alice's preparations by $\rho_x$ which must satisfy the indistinguishability relation $\sum_{r\cdot x=0} \rho_x=\sum_{r\cdot x=1}\rho_x$ for every string $r$ with $|r|\geq 2$. Since one has to keep track of both the statistics and the post-measurement states of each Bob, we require the detailed set of Kraus operators for each measurement. By $K_{y_k}^{b_k}$ we denote the Kraus operators of Bob$_k$ associated to the $y_k$th measurement and $b_k$th outcome. The state received by Bob$_k$ is specified by Alice's input $x$, and the strings of inputs $(y_1,\ldots,y_{k-1})$ and outputs $(b_1,\ldots, b_{k-1})$ of all previous Bobs. However, we treat each Bob in the sequence as independent from the rest, meaning that they do not know the specific inputs or outputs of the other Bobs in each run of the experiment. Thus, in order to calculate the relevant marginal distributions $p(b_k|x,y_k)$, only the average state  $\tilde{\rho}_x^{(k)}$ received by Bob$_k$ is required, i.e.,  the state obtained from averaging a preparation $\rho_x$ of Alice over all the inputs and outputs of all previous Bobs:
\begin{equation}
\tilde{\rho}_x^{(k)}=\frac{1}{n}\sum_{y_{k-1},b_{k-1}} K_{y_{k-1}}^{b_{k-1}}\tilde{\rho}_x^{(k-1)}(K_{y_{k-1}}^{b_{k-1}})^\dagger, 
\end{equation}
with $\tilde{\rho}_x^{(1)}=\rho_x$. Consequently, the desired marginal statistics for Bob$_k$ are $p(b_k|x,y_k)=\Tr\left(\tilde{\rho}_x^{(k)} (K_{y_k}^{b_k})^\dagger K_{y_k}^{b_k}\right)$. This constitutes a description of general quantum strategies in the sequential scenario. 

\subsection{Sharing preparation contextuality} We apply the above general description to construct a specific family of quantum strategies for sharing preparation contextuality, that is inspired by the previously described optimal quantum strategy for the maximal violation of the  inequalities \eqref{ineq}. Alice prepares the states \eqref{optstate} while each Bob performs an unsharp variant of the measurements optimal for violating \eqref{ineq}. In that strategy the measurements of Bob are the dichotomic observables $G_{n,y_k}^T$ defined in \eqref{optmeas}, corresponding to the projectors $\Pi_{n,y}^{b}=(\openone +(-1)^b G_{n,y}^T)/2$ that are both the Kraus operators and POVM elements. For a weaker measurement, one modifies the POVM element to $(\openone +(-1)^b \;\eta_k\; G_{n,y}^T)/2$, for some $\eta_k\in[0,1]$. If $\eta_k=1$ ($\eta_k=0$), the measurement is sharp (non-interacting). Choosing $0<\eta_k<1$ corresponds to an unsharp measurement. The corresponding Kraus operator is given by
\begin{equation}\label{Kraus}
K_{y_k}^{b_k}=\sqrt{\frac{1+\eta_k}{2}} \Pi_{n,y_k}^{b_k}+\sqrt{\frac{1-\eta_k}{2}} \Pi_{n,y_k}^{\bar{b}_k},
\end{equation}
where the bar-sign denotes a bit-flip. This class of strategies has the following convenient property.
\begin{lemma}
If Alice prepares the states in Eq.~\eqref{optstate} and the Bobs each measure $G^T_{n,y_k}$ with sharpness $\eta_k$, the average state received by Bob$_k$ is 
\begin{equation}\label{lemma}
\tilde{\rho}_x^{(k)}=v_k\rho_x+\left(1-v_k\right)\rho_{\text{mix}},
\end{equation}     
where $\rho_{\text{mix}}$ is the maximally mixed state and the visibility $v_k\in [0,1]$ is given recursively by
\begin{equation}\label{visi}
v_k = v_{k-1} f_{k-1} = \prod_{j=1}^{k-1} f_j,
\end{equation}
where $v_1=1$ by definition, and the ``quality factor" $f_k$ of the measurement of Bob$_k$ is defined from the sharpness $\eta_k$ as $f_k = (1+(n-1)\sqrt{1-\eta_k^2})/n$.
\end{lemma}
\begin{proof}
The proof is technical in character and is given in Appendix~\ref{AppendixProofLemma}.
\end{proof}

Using Eq.~\eqref{lemma}, the figure of merit \eqref{ineq} for the pair Alice and Bob$_k$ reads
\begin{equation}\label{figmerit}
\mathcal{A}_k^{(n)}=\frac{1}{2}\left(1+\frac{v_k\eta_k}{\sqrt{n}}\right).
\end{equation}
This leads to preparation contextuality whenever $\eta_k>1/(v_k\sqrt{n})$. This can be used to recursively calculate the critical pairs $(\eta_k,v_{k})$. Thusly, we arrive at the following result. 
\begin{result}\label{res1}
The number of observers who can independently share the preparation contextuality enabled by Alice's ensemble is at least $n$.  
\end{result}
\begin{proof}
Consider that each Bob tunes the sharpness of his measurement so as to just violate the  inequality \eqref{ineq}, but not more. Expressing the measurement sharpness $\eta_k=\sin\theta_k$, where $\theta_k \in [0,\pi/2]$, we thus require $\sin \theta_k = 1/(v_k\sqrt{n})$. On the other hand, a trivial lower bound on the quality factor of Bob$_k$'s measurement is $f_k = \left(1+(n-1)\cos \theta_k\right)/n\geq  \cos \theta_k$. Squaring, and using the expression for the critical value of $\sin\theta_k$ above, we find that $f_k^2\geq 1-1/(v_k^2n)$. Since the visibility of the next Bob is $v_{k+1}=v_kf_k$,  we have $v_{k+1}^2=v_k^2f_k^2\geq v_k^2\left(1-1/(v_k^2n)\right)$. Hence, the decrease in visibility from each Bob to the next is bounded by $v_k^2-v_{k+1}^2\leq 1/n$ which together with $v_1=1$ gives $v_{k+1}^2\geq 1-k/n$. This implies that the visibility of the $n$th Bob is at least $v_n\geq 1/\sqrt{n}$, which is precisely the condition for violating the preparation noncontextuality inequality.
\end{proof}

Thus by suitably choosing $n$, an arbitrary long sequence of observers can share the preparation contextual correlations enabled by Alice's ensemble. Moreover, we  now show that for the considered class of quantum strategies, the number of observers who share preparation contextuality can be no more than $n$.  Consider the quality factor $f_k$ of the the measurement of Bob$_k$. We can find upper and lower bounds on $f_k^2$ in the following manner. First, for a lower bound,
\begin{align}
f_k &= \frac{1 + (n-1) \cos \theta_k}{n} > \cos \theta_k, \\
\therefore f_k^2 &> \cos^2 \theta_k = 1 - \sin^2 \theta_k = 1 - \frac{1}{n v_k^2}.
\end{align}

For the upper bound,
\begin{multline}
f_k^2 < f_k^2 + 4 \frac{(n-1)}{n^2} \sin^4 \frac{\theta_k}{2} \\
= 1 - \frac{(n-1)}{n} \sin^2 \theta_k 
= 1 - \frac{(n-1)}{n^2 v_k^2}.
\end{multline}

But since the visibility $v_{k+1}$ of the next Bob is given by $v_{k+1} = v_k f_k$, we can bound the next visibility as
\begin{align}
v_k^2 \left( 1 - \frac{1}{n v_k^2} \right) < v^2_{k+1} < v_k^2 \left( 1 - \frac{(n-1)}{n^2 v_k^2} \right),
\end{align}
from which the decrease in the visibility squared is both bounded on both sides, by
\begin{align}
\frac{1}{n} < v_k^2 - v_{k+1}^2 < \frac{n-1}{n^2}.
\end{align}

Proceeding from the first Bob, who has visibility $v_1 = 1$, we can use the lower bound to find that $v_n^2 > 1/n$, and the upper bound to find that $v_{n+1}^2 < 1/n$. Since $1/\sqrt{n}$ is the critical visibility to violate the preparation noncontextuality inequality, it follows that Bob$_n$ can violate the inequality (as all of the Bobs before him), but that Bob$_{n+1}$ and later Bobs cannot.

Another noteworthy feature is that one can share preparation contextuality between any number of observers also in a scenario in which none of the Bob's knows his position in the sequence. To this end, consider a quantum strategy in which the set of possible measurements performed by each Bob is the same, i.e., they all perform equally unsharp measurements. If it is the case that the first $k$ Bobs in the sequence violate the preparation noncontextuality inequality, then the weakest violation will be by the last Bob. The condition for the $k$'th Bob to just saturate the preparation noncontextuality inequality reads
\begin{align}
& \sin \theta=\frac{1}{v_{k-1}\sqrt{n}},
\end{align}
where $\eta =\sin \theta$ is the strength of all of the Bobs' measurements, and the visibility is given by,
\begin{equation}
v_{k-1}=\left(\frac{1+(n-1)\cos \theta}{n}\right)^{k-1}.
\end{equation}
Solving the equation for the value of $k$ returns
\begin{equation}\label{eq:temp1}
k=1-\frac{\log \sin \theta+\frac{1}{2}\log n}{\log \left(1+(n-1)\cos\theta\right)-\log n}.
\end{equation}

Consider that the strength of the measurement is chosen to be $\sin\theta = \sqrt{e/n}$ (where $n \geq 3$). In any case, we are interested in the scaling for large $n$, for which we may approximate $\cos \theta = \sqrt{ 1 - e/n } \approx 1 - e/2n$. Substituting this in the above, and further approximating $\log(1-x) \approx -x$ for $|x| \ll 1$, one gets
\begin{align}
k \approx 1 + \frac{n^2}{(n-1)e} = O \left( \frac{n}{e} \right).
\end{align}
One can show that this is the optimal scaling by the Maclaurin expansion of \eqref{eq:temp1} for small $\theta$, and differentiating to find the optimal value of $\theta$.

Thus we find that even in the anonymous setting where each Bob is unaware of their position in the sequence, the maximum number of observers able to share the contextuality enabled by Alice's ensemble by all performing equally unsharp measurements scales as
\begin{equation}
k_{\text{max}}\approx \frac{n}{e}.
\end{equation}
Note that this scaling is the same as obtained in the non-anonymous setting, up to a pre-factor of $1/e$.

\subsection{Noise-robustness} The scenario we have considered so far is an idealisation in which no noise appears. In addition to this not being realistic in any experiment, it is interesting to consider whether the noiseless scenario is distinctive, or also significantly noisy ensembles \cite{footnote} enable shared preparation contextuality. To address this matter, we let Alice's preparations be mixtures of the intended state $\rho_x$ with the maximally mixed state: $\rho_x(q) =q \rho_x+(1-q)\rho_{\text{mix}}$ for some visibility $q\in[0,1]$. For a given number of observers, what is the smallest $q$  such that preparation contextuality can be shared between all observers? 
\begin{result}
	For any given number of independent observers $m$, there exists an ensemble whose contextuality can be shared between all observers for any  $q>0$.
\end{result}
\begin{proof}
The proof follows the same steps as that of Result~\ref{res1}. We substitute $\rho_x$ for $\rho_x(q)$ in the proof of Result~\ref{res1}, which means that  $v_1=q$. Following the same procedure, one obtains the following condition for the visibility in the sequence: $v_{k+1}^2\geq q-k/n$. This implies that to find a violation for the $m$'th Bob, one must choose a sufficiently large $n$, namely $n\geq \lceil \frac{m}{q}\rceil$.
\end{proof}

Hence, preparation contextuality can be shared between any number of observers using ensembles with an arbitrarily large noise-component by choosing a sufficiently large $n$. The price to pay for this property is that when $q\rightarrow 0$, both the Hilbert space dimension of Alice's ensemble and the number of preparations and measurements diverge.

\begin{figure*}
	\includegraphics[width=\textwidth]{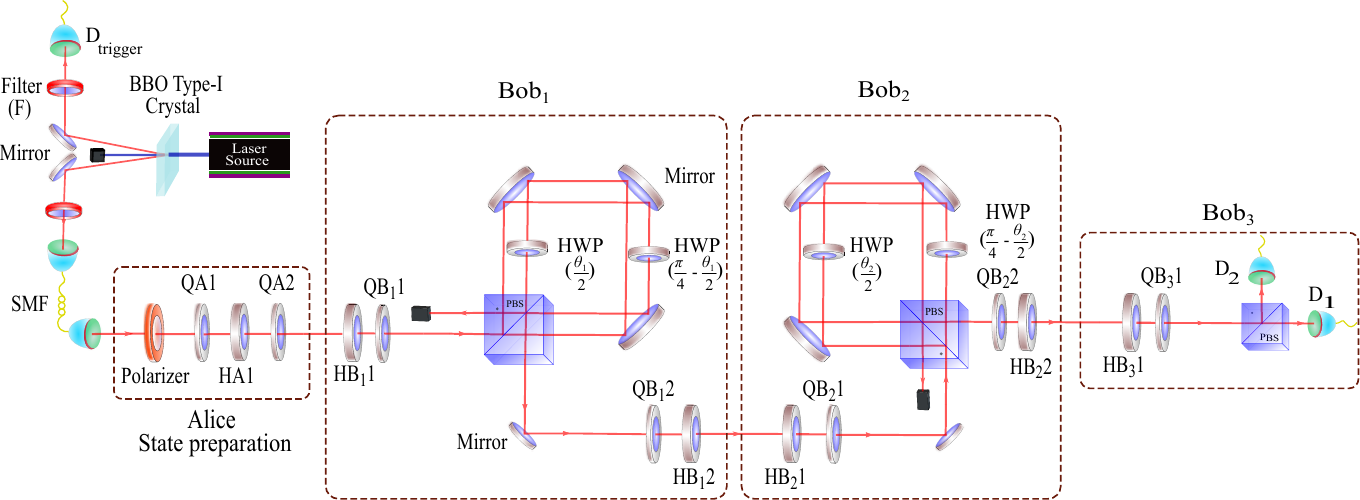}
	\caption{Optical set-up used to reveal contextuality sharing. See text for details. Q and H represent quarter-wave plates (QWPs) and half-wave plates (HWPs).}\label{FigSetup}
\end{figure*}

\section{Experiment}
 We demonstrate the theoretical findings in an experiment with three ($n=3$) sequential tests of preparation contextuality. Alice prepares the eight qubit states \eqref{optstate} with Bloch vectors $\vec{a}_x = \left[ (-1)^{x_1}, (-1)^{x_2},(-1)^{x_3} \right]/\sqrt{3}$. Bob$_1$ and Bob$_2$ perform unsharp  measurements (\ref{Kraus}) of $\sigma_x$, $\sigma_y$ and $\sigma_z$ whereas Bob$_3$ performs projective (sharp) measurements of the same observables.

\par In the experiment we peform unsharp measurements on the polarisation state of a single photon using shifted Sagnac interferometers, as shown in Bob$_1$ and Bob$_2$ in Fig. (\ref{FigSetup}). A HWP is placed in each path of the interferometer, rotated to $\theta_i/2$ in the horizontal path and $\pi/4-\theta_i/2$ in the vertical path to control the sharpness of the measurement. A HWP and QWP before and after the interferometer are used to select the basis of the measurement. The measurement outcome is encoded in the output path, i.e. outcome $b_i=0$ ($b_i=1$) corresponds to the detection of the photon in output path 1 (2, beam blocked in figure). In the sequential scenario we choose to consider only one path at a time for Bob$_1$ and Bob$_2$ to simplify the set-up. By adding an additional rotation to the HWPs or QWPs before and after Bob, we can select the output we want to analyse \cite{Exp1, Exp2}. The results of Bob$_1$ and Bob$_2$'s unsharp measurements are therefore obtained at Bob$_3$, comprised of a PBS and single photon detectors D$_1$ and D$_2$. For example, if we consider output 1 at Bob$_1$ and Bob$_2$, a click in either detector at Bob$_3$ tells us that Bob$_1$ and Bob$_2$ had the outcome $b_1=0$ and $b_2=0$. This is, however, not a limitation since we can equally well have Bob$_1$ and Bob$_2$ read out their respective outputs by placing an additional single photon detector in their respective labs to detect the photon path.  We analyse the counts in Bob$_3$ corresponding to all possible combinations of output ports to realise a full measurement. This protocol relies on a stable photon generation rate. Details of measurement angles are given in Appendix~\ref{ExperimentalDetails}. This set-up can be used to perform projective measurements ($\eta=1$, $\theta_i=0$), no measurement ($\eta=0$, $\theta_i=\pi/4$), or an intermediate-strength measurement, where the the sharpness (strength) of the measurement is tuned  by varying $\theta_i$.

\par The full set-up is shown in Fig.~\ref{FigSetup}. We generate heralded single photons at 780 nm via spontaneous parametric down-conversion (SPDC) using a single type-I beta barium borate (BBO) crystal of thickness 2 mm pumped by 390 nm femto-second laser pulses. The idler photon is detected by an APD single-photon detector, D$_{\mathrm{trigger}}$, and is used as a trigger. The single photons are  coupled into single-mode fibres (SMF) after passing through a narrowband 3 nm interference filter (F) to define the spatial and spectral properties of the  photons. After filtering, the signal photon is prepared into one of Alice's eight states, using a polariser, two QWPs and a HWP (angles given in Appendix~\ref{ExperimentalDetails}). The unsharp measurements of Bob$_1$ and Bob$_2$ correspond to $\theta_1=24.95\degree$ ($\eta_1=0.6441$) and $\theta_2=20.10\degree$ ($\eta_2=0.7637$) respectively, which ideally produce $\mathcal{A}_1=\mathcal{A}_2=\mathcal{A}_3=0.6859>2/3$ with $\mathcal{A}_k=\mathcal{A}_k^{(3)}$.

\subsection{Results}
In order to test each of the three preparation noncontextuality inequalities (between Alice and each of the three Bobs), we require 24 marginal probabilities (the `winning' answers $b_k=x_{y_k}$) corresponding to the three measurement bases and Alice's eight preparations. To reduce the Poissonian error, each Bob collects approximately 34 million counts for each of these 24 settings. Our experimental values can be found in Appendix~\ref{ExperimentalResults}. These lead to three preliminary values of $ \mathcal{A}_1^{\text{pre}}=0.687\pm0.001$, $ \mathcal{A}_2^{\text{pre}}=0.675\pm0.001$, and $\mathcal{A}_3^{\text{pre}}=0.681\pm0.001$.

\subsection{Data analysis} Due to small yet unavoidable experimental imperfections, e.g. waveplate imperfections and offsets in the rotation of the waveplates, it is impossible to perfectly satisfy the operational indistinguishability relations required to test preparation contextuality. This problem can be overcome by suitable post-processing methods \cite{MP16}. As described in Appendix~\ref{AppData}, we have used a relaxed variant of these methods to enforce the indistinguishability relations relevant to a test of inequality \eqref{ineq} on our experimental data. This comes at the cost of the observed values $\left(\mathcal{A}_1^{\text{pre}},\mathcal{A}_2^{\text{pre}},\mathcal{A}_3^{\text{pre}}\right)$ decreasing in a manner corresponding to how well the statistics approximates said relations. Due to the high visibility and precision of the experimental set-up, we find only a small decrease in the three correlation witnesses: 
\begin{align*}
	& \mathcal{A}_1^{\text{post}}=0.683\pm0.001\\
	& \mathcal{A}_2^{\text{post}}=0.670\pm0.001\\
	& \mathcal{A}_3^{\text{post}}=0.677\pm0.001
\end{align*}
all of which violate inequality \eqref{ineq}.

\section{Conclusions} We have theoretically developed and  experimentally demonstrated the sharing of preparation contextual correlations in scenarios that require no entanglement. In addition to such correlations being possible to share between any number of observers, we found that this can be done in a strongly noise-robust manner. This distinguishes shared preparation contextuality from known results in e.g.~shared Bell nonlocality in which the fragility to noise of sequential demonstrations scales super-exponentially \cite{SG16}. This fragility poses a significant experimental hurdle and has hitherto limited demonstrations to two sequential violations of Bell inequalities \cite{Exp1, Exp2}. We experimentally observed three sequential demonstrations of preparation contextuality. Optical set-ups of this spirit (see also Refs~\cite{Exp1, Exp2}) are promising candidates for a variety of sequential correlation tests. Finally, an interesting question is to understand which forms of quantum correlations can be shared between indefinitely many observers in a noise-robust manner.

\textit{Note added.---} In the substantial time between the arXiv submission of this work and its publication, also other works concerning the sequential sharing of contextuality have appeared \cite{Hierarchy2, Pan}.

\section{Acknowledgements}
H. A. and N. W. contributed equally to this work. This work was supported by the project ``Photonic Quantum Information'' (Knut and Alice Wallenberg
Foundation, Sweden), the Swedish Research Council, Stiftelsen Olle Engkvist Byggm{\"a}stare and the Swiss National Science Foundation (Starting grant DIAQ, Early PostDoc Mobility fellowship P2GEP2 194800, and the National Centre of Competence in Research ``Quantum Systems and Technology'').

\appendix
\onecolumngrid

\section{Proof of Lemma}\label{AppendixProofLemma}
In this section, we prove the lemma of the main text. In the considered scenario, Alice receives a random input $x\in\{0,1\}^n$ and prepares the associated state 
\begin{align}
\rho_x=\Tr_\text{A}\left[\left(\openone + A_x\right)\otimes \openone \left(\ketbra{\phi_{\text{max}}}{\phi_{\text{max}}}\right)^{\otimes\lfloor n/2\rfloor}\right],
\end{align}
where $\left(\ketbra{\phi_{\text{max}}}{\phi_{\text{max}}}\right)^{\otimes\lfloor n/2\rfloor}$ is ${\lfloor n/2\rfloor}$ copies of the two-qubit maximally entangled state, and the partial trace is taken over all the first qubits in each pair. Consider that the sequence of Bobs, labelled by $\{1,2,...,m-1\}$, apply measurements of intermediate sharpness to the state above, each denoted by $\eta_k = \sin \theta_k$. We proceed to prove that the average state $\tilde{\rho}_x^{(m)}$ received by Bob$_m$ will be of the form
\begin{align}\label{eq:mthstate}
\tilde{\rho}^{(m)}_x&=\Tr_\text{A}\left[\left(\openone + v_m A_x\right)\otimes \openone \left(\ketbra{\phi_{\text{max}}}{\phi_{\text{max}}}\right)^{\otimes\lfloor n/2\rfloor}\right],
\end{align}
where $v_m$ (the ``visibility" of the state) is given by
\begin{align}
v_m &= v_{m-1} f_{m-1} = \prod_{j=1}^{m-1} f_j, \label{eq:mthvisibility}\\
\text{where} \quad f_j &= \frac{1 + (n-1) \cos \theta_j}{n}. \label{eq:mthquality}
\end{align}
We call $f_j$ the ``quality factor" of the measurement of the $j^{th}$ Bob. The visibility of the first Bob is $v_1 = 1$, since he possesses the undisturbed state received directly from Alice.

The proof is inductive. For the first Bob, the statement holds trivially. Consider that it holds true for $m-1$ Bobs, so that the average state $\tilde{\rho}_x^{(m)}$ received by Bob$_{m}$ is given by \eqref{eq:mthstate}. Then using the Kraus operators stated in the main text, the average state $\tilde{\rho}_x^{(m+1)}$ (averaging over all Bob$_{m}$'s possible and equiprobable inputs, and with no knowledge of his outcome), is given by
\begin{align}
\tilde{\rho}_x^{(m+1)} &= \frac{1}{n}\sum_{y,b}K_y^b \tilde{\rho}_x^{(m)} (K_y^b)^\dagger = \frac{1}{n} \sum_{y,b} \Tr_\text{A} \left[ \left(\openone + v_m A_x\right)\otimes K^b_y \; \left(\ketbra{\phi_{\text{max}}}{\phi_{\text{max}}}\right)^{\otimes\lfloor n/2\rfloor} \; \openone \otimes (K_y^b)^\dagger \right],
\end{align}
where the Kraus operators are acting on the part of the Hilbert space complementary to that being traced out. First, using the property of the maximally entangled state that  $\left( \openone \otimes O \right) \ket{\phi_\text{max}}\bra{\phi_\text{max}} \left( \openone \otimes O^\dagger \right) = \left( O^T \otimes \openone \right) \ket{\phi_\text{max}}\bra{\phi_\text{max}} \left( O^*\otimes \openone\right)$, and then using the cyclicity of the trace, we obtain
\begin{align}
\tilde{\rho}_x^{(m+1)} &= \frac{1}{n} \sum_{y,b} \Tr_\text{A} \left[ \left(\openone + v_m A_x\right) (K^b_y)^T \otimes \openone \; \left(\ketbra{\phi_{\text{max}}}{\phi_{\text{max}}}\right)^{\otimes\lfloor n/2\rfloor} \; (K^b_y)^{\dagger T} \otimes \openone \right] \\
&= \frac{1}{n} \sum_{y,b} \Tr_\text{A} \left[ (K^b_y)^{\dagger T} \left( \openone + v_m A_x \right) (K^b_y)^T \otimes \openone \; \left(\ketbra{\phi_{\text{max}}}{\phi_{\text{max}}}\right)^{\otimes\lfloor n/2\rfloor} \right].\label{eq:blah0}
\end{align}

Splitting the above into the sum of the two terms from the $(\openone + v_m A_x)$, the contribution of the $\openone$ part is
\begin{multline}
\frac{1}{n} \sum_{y,b} \Tr_\text{A} \left[ (K^b_y)^{\dagger T} (K^b_y)^T \otimes \openone \; \left(\ketbra{\phi_{\text{max}}}{\phi_{\text{max}}}\right)^{\otimes\lfloor n/2\rfloor} \right] = \frac{1}{n} \sum_{y} \Tr_\text{A} \left[ \openone \otimes \openone \; \left(\ketbra{\phi_{\text{max}}}{\phi_{\text{max}}}\right)^{\otimes\lfloor n/2\rfloor} \right] \\
= \Tr_\text{A} \left[ \left(\ketbra{\phi_{\text{max}}}{\phi_{\text{max}}}\right)^{\otimes\lfloor n/2\rfloor} \right], \label{eq:blah1}
\end{multline}
where we have used that the $K^b_y$ are Hermitian and that measurements are complete i.e., $\sum_b  (K^b_y)^{\dagger T} (K^b_y)^T = \openone$.

For the term involving $A_x$, we calculate the sum using the Kraus operators from the main text, denoting by $\eta_m = \sin \theta_m$ the strength of the measurement of Bob $m$,
\begin{equation}
K_{y}^{b}=\sqrt{\frac{1+\eta_m}{2}} \Pi_{n,y}^{b}+\sqrt{\frac{1-\eta_m}{2}} \Pi_{n,y}^{\bar{b}} = \left( \frac{\ctm \openone + (-1)^b \stm G^{\red{T}}_{n,y}}{\sqrt{2}} \right),
\end{equation}
which results in
\begin{align}\nonumber
\frac{1}{n} \sum_{y,b} (K^b_y)^{\dagger T} A_x (K^b_y)^T &= \frac{1}{n}\sum_{y,b}  \left( \frac{\ctm \openone + (-1)^b \stm G_{n,y}}{\sqrt{2}} \right) A_x \left( \frac{\ctm \openone + (-1)^b \stm G_{n,y}}{\sqrt{2}} \right) \\\nonumber
&= \frac{1}{2n} \sum_{y,b} \cos^2 \left(\frac{\theta_m}{2}\right) A_x + (-1)^b \cos\left(\frac{\theta_m}{2}\right)\sin \left(\frac{\theta_m}{2}\right) \left\{G_{n,y}, A_x \right\} + \sin^2\left(\frac{\theta_m}{2}\right) G_{n,y} A_x G_{n,y}. \\\nonumber
&= \frac{1}{n} \sum_y \left( \frac{1 + \cos\theta_m}{2} \right) A_x + \left( \frac{1 - \cos\theta_m}{2} \right) G_{n,y} A_x G_{n,y} \\
&= \left( \frac{1 + \cos\theta_m}{2} \right) A_x + \left( \frac{1 - \cos\theta_m}{2} \right) \frac{1}{n} \sum_y G_{n,y} A_x G_{n,y}. \label{eq:blah2}
\end{align}
We may now use the expansion $A_x = \frac{1}{\sqrt{n}} \sum_{i} (-1)^{x_i} G_{n,i}$, and the anti-commutation relation $\{G_{n,j},G_{n,k}\} = 2\delta_{j,k} \openone$ from \cite{IndexGame} to simplify
\begin{align}\nonumber
\frac{1}{n} \sum_y G_{n,y} A_x G_{n,y} &= \frac{1}{\sqrt{n}} \sum_i (-1)^{x_i} \frac{1}{n} \sum_y G_{n,y} G_{n,i} G_{n,y} \\ \nonumber
&= \frac{1}{\sqrt{n}} \sum_i (-1)^{x_i} \frac{1}{n} \sum_y \left( 2\delta_{i,y}G_{n,y}-G_{n,i} \right) \\\nonumber
&= \frac{1}{\sqrt{n}} \sum_i (-1)^{x_i} \frac{1}{n} (2-n) G_{n,i} \\
&= \frac{2-n}{n} A_x.
\end{align}

Inserting this into Eq.~\eqref{eq:blah2}, we obtain
\begin{align}
\frac{1}{n} \sum_{y,b} (K^b_y)^{\dagger T} A_x (K^b_y)^T &= f_m A_x, \\
\text{where} \quad f_m &= \left( \frac{1 + (n-1) \cos \theta_m }{n} \right) = \left( \frac{1 + (n-1) \sqrt{ 1 - \eta^2 } }{n} \right),
\end{align}
is the quality factor of the measurement of Bob$_m$. Combining this with Eq.~\eqref{eq:blah1} to find the final expression for the average state after Bob$_m$'s measuremet Eq.~\eqref{eq:blah0}, we find
\begin{align}
\tilde{\rho}_x^{(m+1)} &= \Tr_\text{A} \left[ \left(\ketbra{\phi_{\text{max}}}{\phi_{\text{max}}}\right)^{\otimes\lfloor n/2\rfloor} \right] + \Tr_\text{A} \left[ v_m f_m A_x \otimes \openone \; \left(\ketbra{\phi_{\text{max}}}{\phi_{\text{max}}}\right)^{\otimes\lfloor n/2\rfloor} \right] \\
&= \Tr_\text{A} \left[ \left( \openone + v_m f_m A_x \right) \otimes \openone \; \left(\ketbra{\phi_{\text{max}}}{\phi_{\text{max}}}\right)^{\otimes\lfloor n/2\rfloor} \right],
\end{align}
which proves the desired relation (\eqref{eq:mthstate} - \eqref{eq:mthquality}).

\section{Experimental settings}\label{ExperimentalDetails}
The angles used for Alice's state preparation are given below:
\begin{table}[h!]
	\centering
	\begin{tabular}{|c|c|c|c|c|}
		\hline
		\bf{State} & \bf{Pol. ($\degree$)} & \bf{QA1 ($\degree$)} & \bf{HA1 ($\degree$)} & \bf{QA2 ($\degree$)} \\ \hline
		000 & 27.37 & 45 & -33.75 & 45 \\ \hline
		001 & 27.37 & 45 & -11.25 & 45 \\ \hline
		010 & 27.37 & 45 & -56.25 & 45 \\ \hline
		011 & 27.37 & 45 & -78.75 & 45 \\ \hline
		100 & 62.63 & 45 & -33.75 & 45 \\ \hline
		101 & 62.63 & 45 & -11.25 & 45 \\ \hline
		110 & 62.63 & 45 & -56.25 & 45 \\ \hline
		111 & 62.63 & 45 & -78.75 & 45 \\
		\hline
	\end{tabular}
	\caption{Angles for the polarizer, QWPs and HWP for the preparation of Alice's states.}
	\label{table:alicesettings}
\end{table}

The settings of the HWPs and QWPs used for the unsharp measurements in Bob$_1$ and Bob$_2$ are as follows. Note these settings are independent of the sharpness of the measurement, which is determined by the angle of the HWPs inside the interferometer.
\begin{table}[h!]
	\centering
	\begin{tabular}{|c|c|c|c|c|c|}
		\hline
		\bf{Measurement} & \bf{Output Port} & \bf{HB$_i$1 ($\degree$)} & \bf{QB$_i$1 ($\degree$)} & \bf{HB$_i$2($\degree$)} & \bf{QB$_i$2 $\degree$)} \\ \hline
		$\sigma_x$ & 1 & 22.5 & 0 & 90 & 22.5 \\ \hline
		$\sigma_x$ & 2 & 67.5 & 0 & 90 & 67.5 \\ \hline
		$\sigma_y$ & 1 & 0 & -45 & 45 & 0 \\ \hline
		$\sigma_y$ & 2 & 0 & 45 & 135 & 0 \\ \hline
		$\sigma_z$ & 1 & 0 & 0 & 90 & 0 \\ \hline
		$\sigma_z$ & 2 & 45 & 0 & 90 & 45 \\
		\hline
	\end{tabular}
	\caption{Waveplate settings for measurement and output selection of Bob$_1$ and Bob$_2$.}
	\label{table:messettings}
\end{table}

\pagebreak
\section{Experimental results}\label{ExperimentalResults}
The experimental marginal probabilities corresponding to the outcomes that satisfy $b_i=x_{y_i}$ (the `winning' answer in the communication game) for Bob$_1$ and Bob$_2$'s unsharp measurements and Bob$_3$'s projective measurements of $\sigma_x$, $\sigma_y$ and $\sigma_z$ on each of Alice's preparations are shown in the following three tables:
\begin{table}[h!]
	\centering
	\begin{tabular}{c|c|c|c|}
		\cline{2-4}
		& \multicolumn{3}{c|}{\bf{Bob$_1$}}\\
		\hline
		\multicolumn{1}{|c|}{\bf{State}} & \bf{$\sigma_x$} & \bf{$\sigma_y$} & \bf{$\sigma_z$} \\ \hline
		\multicolumn{1}{|c|}{000} & $0.7369 \pm 0.0003$ & $0.7044 \pm 0.0003$  & $0.6593 \pm 0.0002$ \\ \hline
		\multicolumn{1}{|c|}{001} & $0.6473 \pm 0.0002$ & $0.7257 \pm 0.0003$  & $0.7079 \pm 0.0003$ \\ \hline
		\multicolumn{1}{|c|}{010} & $0.6900 \pm 0.0003$ & $0.6727 \pm 0.0002$ & $0.6571 \pm 0.0002$ \\ \hline
		\multicolumn{1}{|c|}{011} & $0.6879 \pm 0.0003$ & $0.6501 \pm 0.0002$  & $0.7005 \pm 0.0003$ \\ \hline
		\multicolumn{1}{|c|}{100} & $0.6911 \pm 0.0003$ & $0.6195 \pm 0.0002$  & $0.7180 \pm 0.0003$ \\ \hline
		\multicolumn{1}{|c|}{101} & $0.6813 \pm 0.0003$ & $0.6464 \pm 0.0002$  & $0.6779 \pm 0.0003$ \\ \hline
		\multicolumn{1}{|c|}{110} & $0.6400 \pm 0.0002$ & $0.7471 \pm 0.0003$  & $0.7125 \pm 0.0003$ \\ \hline
		\multicolumn{1}{|c|}{111} & $0.7242 \pm 0.0003$ & $0.7132 \pm 0.0003$  & $0.6755 \pm 0.0002$ \\
		\hline
	\end{tabular}
	\caption{Experimental marginal probabilities for Bob$_1$.}
	\label{table:results1}
\end{table}

\begin{table}[h!]
	\centering
	\begin{tabular}{c|c|c|c|}
		\cline{2-4}
		& \multicolumn{3}{|c|}{\bf{Bob$_2$}}\\
		\hline
		\multicolumn{1}{|c|}{\bf{State}} & \bf{$\sigma_x$} & \bf{$\sigma_y$} & \bf{$\sigma_z$} \\ \hline
		\multicolumn{1}{|c|}{000} & $0.6997 \pm 0.0003$ & $0.6422 \pm 0.0002$ & $0.6851\pm 0.0003$ \\ \hline
		\multicolumn{1}{|c|}{001} & $0.6586 \pm 0.0002$ & $0.6785 \pm 0.0002$ & $0.6746\pm 0.0002$ \\ \hline
		\multicolumn{1}{|c|}{010} & $0.6537 \pm 0.0002$ & $0.7088 \pm 0.0003$ & $0.6715\pm 0.0002$ \\ \hline
		\multicolumn{1}{|c|}{011} & $0.6896 \pm 0.0003$ & $0.6824 \pm 0.0003$ & $0.6572\pm 0.0002$ \\ \hline
		\multicolumn{1}{|c|}{100} & $0.7106 \pm 0.0003$ & $0.6370 \pm 0.0002$ & $0.6775\pm 0.0002$ \\ \hline
		\multicolumn{1}{|c|}{101} & $0.6446 \pm 0.0002$ & $0.6792 \pm 0.0003$ & $0.6868\pm 0.0003$ \\ \hline
		\multicolumn{1}{|c|}{110} & $0.6553 \pm 0.0002$ & $0.7000 \pm 0.0003$ & $0.6752\pm 0.0002$ \\ \hline
		\multicolumn{1}{|c|}{111} & $0.6787 \pm 0.0002$ & $0.6666 \pm 0.0002$ & $0.6853\pm 0.0002$ \\
		\hline
	\end{tabular}
	\caption{Experimental marginal probabilities for Bob$_2$.}
	\label{table:results2}
\end{table}

\begin{table}[h!]
	\centering
	\begin{tabular}{c|c|c|c|}
		\cline{2-4}
		& \multicolumn{3}{|c|}{\bf{Bob$_3$}}\\
		\hline
		\multicolumn{1}{|c|}{\bf{State}} & \bf{$\sigma_x$} & \bf{$\sigma_y$} & \bf{$\sigma_z$} \\ \hline
		\multicolumn{1}{|c|}{000} & $0.7044\pm 0.0003$  & $0.6470\pm 0.0002$  & $0.6786\pm 0.0003$  \\ \hline
		\multicolumn{1}{|c|}{001} & $0.6661\pm 0.0002$  & $0.6886\pm 0.0003$  & $0.6854\pm 0.0003$  \\ \hline
		\multicolumn{1}{|c|}{010} & $0.6582\pm 0.0002$  & $0.7113\pm 0.0003$  & $0.6702\pm 0.0002$  \\ \hline
		\multicolumn{1}{|c|}{011} & $0.7011\pm 0.0003$  & $0.6783\pm 0.0003$  & $0.6746\pm 0.0003$  \\ \hline
		\multicolumn{1}{|c|}{100} & $0.6975\pm 0.0003$  & $0.6558\pm 0.0002$  & $0.6915\pm 0.0003$  \\ \hline
		\multicolumn{1}{|c|}{101} & $0.6655\pm 0.0003$  & $0.7049\pm 0.0003$  & $0.6891\pm 0.0003$  \\ \hline
		\multicolumn{1}{|c|}{110} & $0.6469\pm 0.0002$  & $0.6942\pm 0.0003$  & $0.6881\pm 0.0003$  \\ \hline
		\multicolumn{1}{|c|}{111} & $0.7027\pm 0.0003$  & $0.6512\pm 0.0002$  & $0.6853\pm 0.0003$  \\
		\hline
	\end{tabular}
	\caption{Experimental marginal probabilities for Bob$_3$.}
	\label{table:results3}
\end{table}

\pagebreak
\section{Enforcing strict operational equivalences on experimental data}\label{AppData}
Tests of preparation contextuality require that the observed probabilities satisfy an equivalence relation. In the specific preparation noncontextuality inequalities considered in the main text, that equivalence relation follows from the indistinguishability relation imposed on Alice's quantum preparations, i.e. that she hides  the value of the parity $r\cdot x$ for every string $r\in\{0,1\}^n$ with $|r|\geq 2$. This is an operational equivalence relation that is expressed in terms of probabilities as follows,
\begin{equation}\label{opeq}
\forall r, \forall M: \sum_{r\cdot x=0}p(P_x|b,M)=\sum_{r\cdot x=1}p(P_x|b,M).
\end{equation}
Evidently, due to unavoidable experimental imperfections, such a constraint can never be exactly satisfied. This necessitates data processing methods to contend with the problem. Ref.~\cite{MP16} developed a method for post-processing outcome statistics that approximately satisfies an operational equivalence constraint into data that strictly satisfies said constraint. The price to pay for this mapping is that the value of the witness after post-processing is worse than what is originally measured. Roughly speaking, the closer the unprocessed outcome statistics is to satisfying the operational equivalence constraint, the smaller the decrease in the witness value due to the post-processing scheme. 

We have applied a simplified variant (which assumes that the experiment is accurately described by quantum theory) of the method of \cite{MP16} to enforce operational equivalence in each of the three sequential tests of preparation contextuality. We describe how it applies to the experimental results of the pair Alice-Bob$_1$. Since the the outcomes are binary, the full distribution $p(b_1|x,y_1)$ can be described by only considering $p(b_1=0|x,y_1)$. We can write this distribution as eight vectors $\mathbf{P}_x=[p(0|x,1), p(0|x,2), p(0|x,3)]$. The vectors $\mathbf{P}_x$ will not perfectly satisfy the operational equivalence constraint \eqref{opeq}. Therefore, we aim to map them to other distributions $\mathbf{P}'_x$ which perfectly satisfy \eqref{opeq}. This can be done by noting that an experiment in which $\{\mathbf{P}_x\}$ is realised, also constitutes an effective realisation of all distributions in the convex hull of $\{\mathbf{P}_x\}$ (due to linearity).  Hence, we set 
\begin{equation}
\mathbf{P}'_x=\sum_{x'} \omega_{x}^{x'}\mathbf{P}_{x'},
\end{equation}
where for $\forall x$ $\{\omega_x^{x'}\}_{x'}$ is a probability distribution. We search a set of distributions $\{\omega_x\}$ that maximises the witness of preparation contextuality while also enforcing \eqref{opeq}. This problem is solved with a linear program 
\begin{align}
&\qquad\qquad  \mathcal{A}_1^{\text{post}}=\max_{\{\omega\}} \mathcal{A}_1^{\text{pre}}[\{\mathbf{P}_x'\}] & \text{such that } \forall r\in\{011,101,110,111\} \qquad \sum_{r\cdot x=0}\mathbf{P}'_{x}=\sum_{r\cdot x=1}\mathbf{P}'_{x}.
\end{align}
In addition, we can employ the quantity $F=\sum_x \omega_{x}^x$ as a measure of the closeness of the observed and post-processed data. Moreover, this procedure can be straightforwardly adapted to the experimental results obtained for Alice-Bob$_2$ and Alice-Bob$_3$. The minor difference is that the preparation procedure for e.g.~Alice-Bob$_2$ effectively becomes the average state relayed by Bob$_1$ to Bob$_2$. Thus, change the definition of vectors $\mathbf{P}_x$ to instead apply to the distributions $p(b_2=0|x,y_2)$ and $p(b_3=0|x,y_3)$ respectively and proceed in analogy with the above. 

Solving the above linear program, we have obtained the following results for the three demonstrations of preparation contextuality.
\begin{align}
& \mathcal{A}_1^{\text{pre}}=0.687 & \mathcal{A}_1^\text{post}=0.683 && F=0.9690\\
& \mathcal{A}_2^{\text{pre}}=0.675 & \mathcal{A}_2^\text{post}=0.670  &&  F=0.9537\\
& \mathcal{A}_3^{\text{pre}}=0.681 & \mathcal{A}_3^\text{post}=0.677 && F=0.9700
\end{align}


\begin{thebibliography}{99}
\bibitem{footnote}
One could alternatively consider the Bobs' measurement devices inducing the noise. However, this is less detrimental than noisy preparations. The reason is that if Alice's preparations are noisy the correlations due to all Bobs' measurements are weaker, whereas if instead one (or many) of the Bobs sometimes fail to perform the intended measurement, the state relayed to the next Bob retains a higher degree of coherence and leads to him observering stronger correlations.
\end{thebibliography}
\end{document}